\documentclass[11pt]{article}

\usepackage{fullpage}
\usepackage{here}
\usepackage{amsthm,amsmath,amssymb}
\usepackage{amsfonts}
\usepackage{yhmath}
\usepackage{xcolor}
\usepackage{graphicx}
\usepackage{booktabs}
\usepackage{caption}
\usepackage{hyperref}
\usepackage{enumerate}
\usepackage{enumitem}
\usepackage{wrapfig}

\theoremstyle{plain}
\newtheorem{theorem}{Theorem}
\newtheorem{lemma}[theorem]{Lemma}

\newtheorem{corollary}[theorem]{Corollary}

\newtheorem{proposition}[theorem]{Proposition}
\newtheorem{problem}[theorem]{Problem}

\begin{document}
\title{Plane Multigraphs with One-Bend and\\ Circular-Arc Edges of a Fixed Angle\thanks{Research on this paper was partially supported by the NSF award DMS~2154347.}}
\author{Csaba D. T\'oth\thanks{California State University Northridge, Los Angeles, CA, USA; and 
Tufts University, Medford, MA, USA. Email: \texttt{csaba.toth@csun.edu}}}
\date{}

%
\maketitle

\begin{abstract}
For an angle $\alpha\in (0,\pi)$, we consider plane graphs and multigraphs in which the edges are either (i) one-bend polylines with an angle $\alpha$ between the two edge segments, or (ii) circular arcs of central angle $2(\pi-\alpha)$. We derive upper and lower bounds on the maximum density of such graphs in terms of $\alpha$. As an application, we improve upon bounds for the number of edges in $\alpha AC_1^=$ graphs (i.e., graphs that can be drawn in the plane with one-bend edges such that any two crossing edges meet at angle $\alpha$). This is the first improvement on the size of $\alpha AC_1^=$ graphs in over a decade.
%
%\smallskip\noindent\textbf{keywords:} circular arc, one-bend drawing, $\alpha$-angle crossing drawing
\end{abstract}

\section{Introduction}
\label{sec:intro}

According to a well-known corollary of Euler's formula, an edge-maximal planar straight-line graph on $n\geq 3$ vertices has at most $3n-6$ edges, which is attained on any set of $n\geq 3$ points in general position with a triangular convex hull; and at least $2n-3$ edges, which is attained for $n$ points in convex position. Specifically, on a given set $P$ of $n$ points in the plane in general position, $h\geq 3$ of which are on the convex hull of $P$, the maximum number of edges of a planar straight-line graph is $M(P)=3n-h-3$. This paper explores analogous questions for graphs where the edges are one-bend polylines or circular arcs with a fixed angle. Importantly, there may be multiple edges between a pair of vertices in these drawing styles, and \emph{multigraphs} become relevant.

\begin{wrapfigure}[10]{r}{0.35\textwidth}
  \vspace{-1.5\baselineskip}
  \begin{center}
   \includegraphics[width=0.3\textwidth]{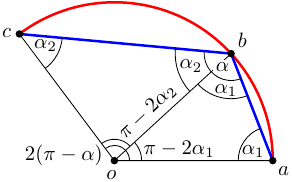}
  \end{center}
    \vspace{-1\baselineskip}
\caption{The relation between an $\alpha$-arc edge and an $\alpha$-bend edges.} \label{fig:0}
\end{wrapfigure}
%
%\begin{figure}[htbp]
%\centering
%\includegraphics[width=.36\textwidth]{arc-bend}
%\caption{The relation between an $\alpha$-arc edge and an $\alpha$-bend edges.} \label{fig:0}
%\end{figure}
%
For an angle $\alpha\in (0,\pi)$, an \textbf{$\alpha$-bend edge} between vertices $a$ and $c$ is a polyline $(a,b,c)$ with one bend at $b$ such that the interior angle of the triangle $\Delta (abc)$ at $b$ is $\alpha$; and an \textbf{$\alpha$-arc edge} is a circular arc between $a$ and $c$ with central angle $2(\pi-\alpha)$. By the Inscribed Angle Theorem, if $e$ is an $\alpha$-arc edge between $a$ and $c$, and $b$ is any interior point of the arc $e$, then the polyline $(a,b,c)$ is an $\alpha$-edge; see Fig.~\ref{fig:0} for an example.

A simple graph embedded\footnote{An \emph{embeddng} of a graph into a surface is a continuous injective map of the 1-dimensional simplicial complex formed by the vertices and edges of the graph.} in the plane such that every edge is $\alpha$-bend (resp., $\alpha$-arc) is an \textbf{$\alpha$-bend graph} (resp., \textbf{$\alpha$-arc graph}).
Similarly, a multigraph embedded in the plane such that every edge is $\alpha$-bend (resp., $\alpha$-arc) is an \textbf{$\alpha$-bend multigraph} (resp., \textbf{$\alpha$-arc multigraph}). See Fig.~\ref{fig:1} for examples.
For a finite set $P\subset \mathbb{R}^2$, denote by $M_a(P,\alpha)$ and $M_b(P,\alpha)$, resp., the maximum number of edges in an $\alpha$-arc graph and an $\alpha$-bend graph. Similarly, the
maximum number of edges in an $\alpha$-arc and $\alpha$-bend multigraph is denoted by $M_a^\|(P,\alpha)$ and $M_b^\|(P,\alpha)$, respectively.
It may be hard to compute $M_a(P,\alpha)$, $M_b(P,\alpha)$, $M_a^\|(P,\alpha)$, and $M_b^\|(P,\alpha)$, for a given point set $P$ and a given angle $\alpha\in (0,\pi)$; see Problem~\ref{prob:1} in Section~\ref{sec:con}.

\begin{figure}[htbp]
\centering
\includegraphics[width=0.9\textwidth]{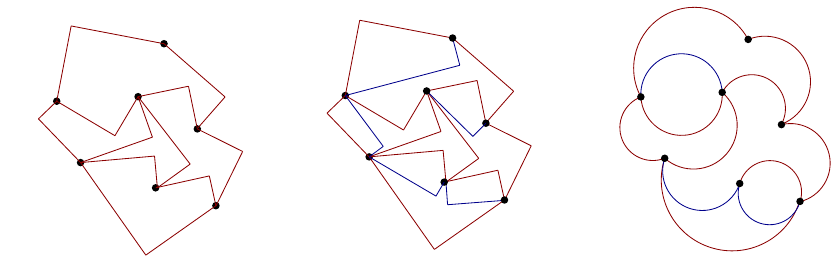}
\caption{A $\frac{\pi}{2}$-bend graph,
a $\frac{\pi}{2}$-bend multigraph,
and a $\frac{\pi}{2}$-arc muligraph} \label{fig:1}
\end{figure}

For $n\in \mathbb{N}$, let $M_a(n,\alpha) = \sup_{|P|=n} M_a(P,\alpha)$,
and define $M_b(n,\alpha)$, $M_a^\|(n,\alpha)$ and $M_b^\|(n,\alpha)$ analogously.
Our bounds on these quantities are in Table~\ref{table:1}.

\begin{table}[ht]
	\caption{Overview of results for $n\geq 3$ (without assuming general position). \label{table:1}}
\centering
	\begin{tabular}{|l|c|c|c|c|c|c|}
\hline
 Angle     & $\alpha\to 0$ & $\alpha\in (0,\frac{\pi}{2}]$ & $\alpha\in (\frac{\pi}{2},\frac{2\pi}{3}]$ & $\alpha\in [\frac{5\pi}{5},\frac{5\pi}{6}]$ &$\alpha\in (\frac{2\pi}{3},\pi)$ & $\alpha\to \pi$ \\ \hline
$M_a(n,\alpha)$ & $2n-3$ & $2n-3\leq . \leq 3n-6$ & $3n-6$ & $3n-6$ &$3n-6$ & $3n-6$\\ \hline
$M_b(n,\alpha)$ & $3n-6$ & $3n-6$ & $3n-6$  & $3n-6$ & $3n-6$ & $3n-6$\\ \hline	
$M_a^\|(n,\alpha)$ & $2n-2$ & $2n-2\leq . \leq 4n-6$ & $4n-6$ & $4n-6$ &  $6n-O(\sqrt{n})$ & $6n-12$\\ \hline	
$M_b^\|(n,\alpha)$ & $4n-6$ & $4n-6$ & $4n-6$ & $6n-O(\sqrt{n})$ & $6n-O(\sqrt{n})$ & $6n-12$\\ \hline		
	\end{tabular}
\end{table}

\paragraph{Motivation: RAC and $\alpha AC^=$ Drawings.}
A \textbf{right angle crossing drawing} (or \emph{RAC} drawing, for short) of a graph $G=(V,E)$
is a drawing in which edges are polylines and crossing edges meet at angle $\frac{\pi}{2}$.
A \textbf{RAC$_b$} drawing is a RAC drawing where every edge is drawn as a polyline with $b$ bends;
and a graph $G=(V,E)$ is a \textbf{RAC$_b$ graph} if it admits such a drawing.
Didimo et al.~\cite{DidimoEL11} proved that a RAC$_0$-graph on $n\geq 3$ vertices has at most $4n-10$ edges,
and this bound is the best possible; see also~\cite{DujmovicGMW11}.
They also showed that every graph is a RAC$_3$ graph.
Angelini et al.~\cite{AngeliniBFK20} proved that an $n$-vertex RAC$_1$ graph has at most $5.5n-O(1)$ edges,
and this bound is asymptotically tight. Arikushi et al.~\cite{ArikushiFKMT12} showed that an $n$-vertex RAC$_2$ graph has at most $74.2n$ edges, which was recently improved to $20n$ \cite{Toth2023rac}, and is conjectured to be $10n-O(1)$~\cite{AgeliniBK00U23}.
Refer to the surveys~\cite{Didimo20,DidimoLM19} for an overview on RAC drawings and their relatives.

Dujmovi\'{c} et al.~\cite{DujmovicGMW11} extended the notion of RAC drawings to drawings where the crossing edges meet at an angle greater than some $\alpha$, for $\alpha\in (0,\frac{\pi}{2})$, and call such drawings \textbf{$\alpha$ angle crossing drawings} (\textbf{$\alpha AC$ drawings}, for short). They proved that an $n$-vertex graph with a straight-line $\alpha AC$ drawing has at most $\frac{\pi}{\alpha}(3n-6)$ edges.
Ackerman et al.~\cite{AckermanFT12} defined \textbf{$\alpha AC^=_b$ graphs}, which are graphs that can be drawn such that the edges are polylines with at most $b$ bends per edge and every crossing occurs exactly at angle $\alpha$. They proved that $n$-vertex $\alpha AC^=_1$ and $\alpha AC^=_2$ graphs have $O(n)$ edges for any $\alpha\in (0,\frac{\pi}{2}]$. In an $\alpha AC^=_1$ drawing, each edge is a polyline with two \textbf{segments} that are also called \textbf{end-segments}.
In an $\alpha AC^=_2$ drawing, each edge is a polyline with three segments: Two \textbf{end-segments} incident to the vertices, and one \textbf{middle segment}.

\begin{figure}
\centering
\includegraphics[width=0.7\textwidth]{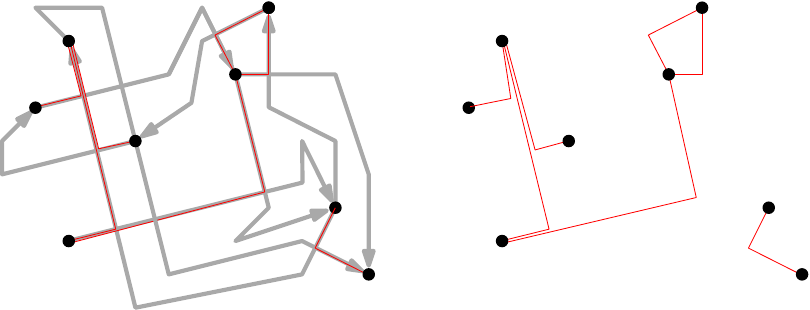}
\caption{Left: A RAC$_2$ drawing of a graph $G=(V,E)$ and the red graph $(V,\Gamma)$.
Right: Perturbation of overlapping red edges yields a $\frac{\pi}{2}$-bend multigraph.} \label{fig:2}
\end{figure}

The main technical tool in the proofs by Ackerman et al.~\cite{AckermanFT12} were $\alpha$-bend graphs, although they did not use this terminology. For example, suppose that we are given an $\alpha AC^=_2$-drawing of a directed graph $G=(V,E)$ in which the edge directions determine a \emph{first} and a \emph{last} end-segment, and suppose that all crossings are between first and last end-segments; see Fig.~\ref{fig:2}. We can create the multigraph $(V,\Gamma)$, called the \emph{red graph}, as follows:
If the first end-segment $s$ of some edge crosses any other edge, then we create an edge $\gamma(s)\in \Gamma$ as a directed path in the planarization of the $\alpha AC^=_2$-drawing: The path $\gamma(s)$ starts from the (unique) vertex in $V$ incident to $s$, then follows $s$ until its first crossing with the end-segment $s'$ of some other edge, and then it follows $s'$ to the (unique) vertex in $V$ incident to $s'$. Each edge of the red graph is an $\alpha$-bend or a $(\pi-\alpha)$-bend edge. The edges of $(V,\Gamma)$ do not cross---they may partially overlap, but they can be perturbed to remove overlaps. Thus $(V,\Gamma)$ is the union of two plane multigraphs: an $\alpha$-bend and a $(\pi-\alpha)$-bend multigraph.

Ackerman et al.~\cite{AckermanFT12} proved that every $n$-vertex $\alpha AC_1^=$-graph has at most $27n$ edges for $\alpha\in (0,\frac{\pi}{2}]$. We improve this bound to $21n-36$ for $n\geq 3$ (Theorem~\ref{thm:AC}) using the bound $M_b^\|(n,\alpha)\leq 4n-6$ for $\alpha\in (0,\frac{2\pi}{3}]$ (Theorem~\ref{thm:recursion}), applied to a red $\alpha$-bend or $(\pi-\alpha)$-bend multigraph.
This is the first improvement on the size of $\alpha AC_1^=$-graphs in more than a decade.

\paragraph{Further Related Previous Work.}
Angle constraints in graph drawing have been considered since the 1980s. Vijayan~\cite{Vijayan86} introduced \emph{angle graphs}, which are graphs such that at every vertex, the rotation of incident edges as well as the angles between consecutive edges in the rotation are given, and asked whether a given angle graph can be realized by a straight-line drawing (possibly with crossings). Planarity testing for angle graphs is NP-hard~\cite{BekosFK19,Garg98}, but there is a linear-time algorithm for triangulations (with a triangular outer face)~\cite{Battista96}. Note that a realization of an angle graph may have crossing edges: Efrat et al.~\cite{EfratFKT22} showed that for cycles, one can find a realization with the minimum number of crossings. Importantly, the edges of an angle graph are realized by straight-line segments. In contrast, we consider graphs with one-bend or circular arc edges of a fixed angle, however the angles between adjacent edges are unconstrained.

Circular arcs and one-bend polylines are among the most popular graph drawing styles. In general, both drawing styles allow more flexibility than straight-line drawings. For example, there are universal point sets of size $O(n)$ for planar $n$-vertex graphs if the edges are drawn as circular arcs~\cite{AngeliniEFKLMTW14} or as one-bend polylines~\cite{EverettLLW10}, even if the bend points are restricted to the universal point set~\cite{LofflerT15},
while the current best universal point set for straight-line embeddings is of size $n^2/2$~\cite{BannisterCDE14}. Chaplick et al.~\cite{ChaplickFK020} showed that an $n$-vertex RAC drawing with circular arc edges can have $4.5n-O(\sqrt{n})$ edges, as opposed to at most $4n-10$ edges in a straight-line RAC drawing.
Refer to the surveys~\cite{Didimo20,DidimoLM19} for a variety of results on RAC drawings.

One-bend polylines and circular arcs lose most of their competitive advantage against straight-line drawings if the angle $\alpha$ is fixed. A small angle $\alpha>0$ may even be a disadvantage. Our objective is to obtain quantitative bounds to compare planar straight-line graphs to $\alpha$-bend and $\alpha$-arc graphs and multigraphs.

\section{Preliminaries}
\label{sec:pre}

We show that the multiplicity of any edge in a $\alpha$-bend and an $\alpha$-arc graph is at most two.

\begin{proposition}
\label{pro:multiplicity}
For every $\alpha\in (0,\pi)$ and every pair of points $a,c\in \mathbb{R}^2$, a plane $\alpha$-bend (resp., $\alpha$-arc) multigraph contains at most two edges between $a$ and $c$, at most one in each halfplane bounded by the line $ac$.
\end{proposition}
\begin{proof}
Let $\alpha\in (0,\pi)$, and let $a$ and $c$ be distinct points in the plane. There are precisely two circular arcs of central angle $\pi-\alpha$ between $a$ and $c$, which lie in distinct halfplanes bounded by the line $ac$. Hence there are at most two $\alpha$-arc edges between $a$ and $c$. Every $\alpha$-bend edge between $a$ and $c$ is a polyline $(a,b,c)$ where the bend pointy $b$ is on an $\alpha$-arc edge between $a$ and $c$. However, if $(a,b_1,c)$ and $(a,b_2,c)$ are $\alpha$-bend edges, then $b_1$ and $b_2$ lie on distinct $\alpha$-arcs between $a$ and $c$, otherwise the two $\alpha$-bend edges would cross: Indeed, assume w.lo.g.\ that the points $a,b_1,b_2,c$ are in this order along the same circular arc from $a$ to $c$. Then the line segments $ab_2$ and $b_1c$ cross, hence the two $\alpha$-bend edges cross.
\end{proof}

It is easy to find the maximum number of edges for collinear points.

\begin{proposition}\label{pro:collinear}
For every set $P$ of $n\geq 3$ collinear points, we have
\begin{enumerate}\itemsep 0pt
\item $M_b(P,\alpha)=3n-6$ and $M_b^\|(P,\alpha)=4n-6$ for all $\alpha\in (0,\pi)$;
\item $M_a(P,\alpha)=3n-6$ and $M_a^\|(P,\alpha)=4n-6$ for all $\alpha\in [\frac{\pi}{2},\pi)$;
\item $M_a(P,\alpha)=2n-3$ and $M_a^\|(P,\alpha)=2n-2$ for all $\alpha\in (0,\frac{\pi}{2})$.
\end{enumerate}
\end{proposition}
\begin{proof}
Assume w.l.o.g.\ that $P=\{p_1,\ldots , p_n\}$ is a set of $n$ points on the $x$-axis sorted by increasing $x$-coordinates.

\smallskip\noindent\textbf{Upper Bounds.}
By definition, $\alpha$-arc and $\alpha$-bend graphs are simple planar graphs. By Euler's polyhedron formula, $M_a(P,\alpha)\leq 3n-6$ and $M_b(P,\alpha)\leq 3n-6$. Suppose that multiple edges are allowed. By Proposition~\ref{pro:multiplicity}, the edges in each halfplane form a simple graph. Since all vertices are on the boundary of both halfplanes, the edges in each halfplane form an outerplanar graph each with at most $2n-3$ edges. This proves $M_a^\|(P,\alpha)\leq 4n-6$ and $M_b^\|(P,\alpha)\leq 4n-6$ for all $\alpha\in (0,\pi)$.

It remains to consider $\alpha$-arc graphs and multigraphs for $\alpha\in (0,\frac{\pi}{2})$.
Recall that an \textbf{$\alpha$-arc edge} is a circular arc between two points with central angle $2(\pi-\alpha)$.
Note that $2(\pi-\alpha)\in (\pi,2\pi)$ for $\alpha\in (0,\frac{\pi}{2})$, thus an $\alpha$-arc edge is more than a halfcircle.
Let $G=(V,E)$ be an $\alpha$-arc multigraph, and let $G^+=(V,E^+)$ and $G^-=(V,E^-)$ be the subgraphs formed by the edges in the upper and lower halfplane, respectively. If $i<j<k$ and $\alpha\in (0,\frac{\pi}{2})$, then the $\alpha$-arc edges $p_ip_j$ and $p_jp_k$ would cross if they are both in the upper halfplane (or both in the lower halfplane). Consequently, in both $G^+$ and $G^-$, each vertex is either to the left or to the right of all of its neighbors.
It follows that neither $G^+$ nor $G^-$ can contain cycles, hence they each have at most $n-1$ edges. This implies $M_a^\|(P,\alpha)\leq 2n-2$ for all $\alpha\in (0,\frac{\pi}{2})$. Note that if $G$ is a simple graph, then $G^+$ and $G^-$ cannot both contain the edge $p_1p_n$. Assume w.l.o.g.\ that $G^-$ does not contain the edge $p_1p_n$. Then $G^-$ cannot contain any path from $p_1$ to $p_n$, consequently it is a forest with at least two components, and so it has at most $n-2$ edges.
This yields the upper bound $M_a(P,\alpha)=2n-3$ for $\alpha\in (0,\frac{\pi}{2})$.

\smallskip\noindent\textbf{Lower Bound Constructions.}
We start with $\alpha$-arc graphs. For $\alpha\in [\frac{\pi}{2},\pi)$, an $\alpha$-arc edge $p_ip_j$ is an $x$-monotone arc that lies in the vertical strip bounded by the vertical lines through $p_i$ and $p_j$.
Add $\alpha$-arc edges $p_i p_{i+1}$ for all $i=1,\ldots ,n-1$ and edges $p_1 p_j$ for all $j=2,\ldots n$ in the upper halfplane; and edge $p_jp_n$ for all $j=2,\ldots ,n-2$ in the lower halfplane. This yields $(n-1)+(n-2)+(n-3)=3n-6$ edges; see the red edges in Fig.~\ref{fig:3}~(left).
We can augment this construction to an $\alpha$-arc multigraph by adding $\alpha$-arc edges $p_ip_{i+1}$ for all $i=1,\ldots , n-1$ and the edge $p_ip_n$ in the lower halfplane as well. This yields an $\alpha$-arc multipraph with $(3n-6)+n=4n-6$ edges; see Fig.~\ref{fig:3}~(left).

For $\alpha\in (0,\frac{\pi}{2})$, we construct an $\alpha$-arc graph with $2n-3$ edges as follows. In the upper halfplane, we use the star formed by $p_1 p_i$ for all $i=2,\ldots ,n$; and in the lower halfplane, the star formed by $p_jp_n$ for all $j=2,\ldots ,n-2$. For an $\alpha$-arc multigraph with $2n-4$ edges, we add the edge $p_1p_n$ in the lower halfplane as well. Both are subgraphs of the constructions in Fig.~\ref{fig:3}~(left).

\begin{figure}
\centering
\includegraphics[width=0.7\textwidth]{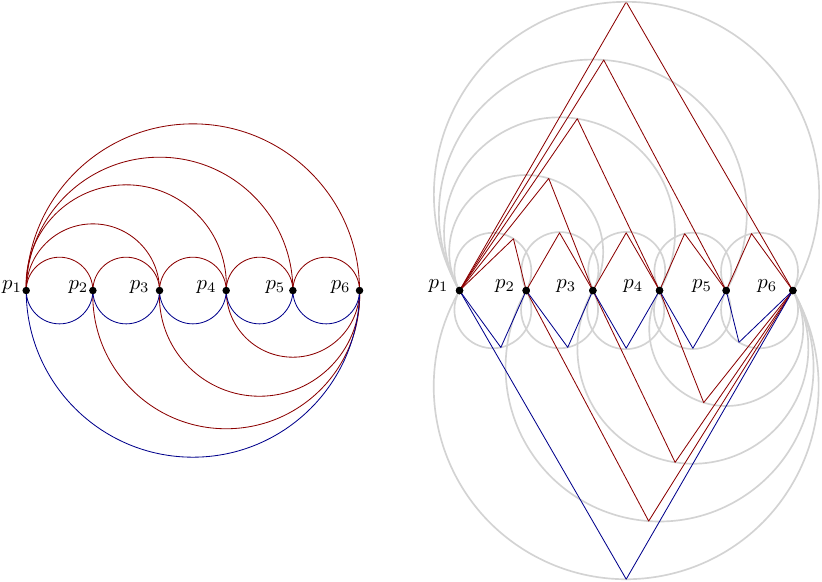}
\caption{Left: a $\frac{\pi}{2}$-arc multigraph with $4n-6$ edges; the red edges form a $\frac{\pi}{2}$-arc graph with $3n-6$ edges.
Right: a $\frac{\pi}{3}$-bend multigraph with $4n-6$ edges; the red edges form a $\frac{\pi}{3}$-bend graph with $3n-6$ edges. In both examples, $n=6$.} \label{fig:3}
\end{figure}

We can construct an $\alpha$-bend graph with $3n-6$ edges and an $\alpha$-bend multigraph with $4n-6$ edges by connecting the same pair of vertices for any $\alpha\in (0,\pi)$. We describe the construction is three steps: (1) Draw all $\alpha$-arc edges described above (these edges may cross for $\alpha<\pi/2$). (2) For each $\alpha$-arc $p_ip_j$, create an $\alpha$-bend polyline $(p_i,b,p_j)$ such that $b$ is the midpoint of the $\alpha$-arc, which implies that $\Delta(p_i b p_j)$ is an isosceles triangle. Such an $\alpha$-bend edge $p_ip_j$ is an $x$-monotone arc for any $\alpha\in (0,\pi)$. Consequently, these $\alpha$-bend edges do not cross, but adjacent edges may partially overlap. (3) We successively perturb the $\alpha$-bend edges as follows; see Fig.~\ref{fig:3}~(right). We perform the perturbation in each halfplane independently. Consider the edges $p_ip_j$ in the upper (resp., lower) halfplane ordered by nonincreasing length. For an edge $p_ip_k$, if both $p_ip_j$ and $p_jp_k$ are edges for some $i<j<k$, then perturb both $p_ip_j$ and $p_jp_k$ by slightly moving their bend points along the corresponding $\alpha$-arc towards each other (i.e., counterclockwise and clockwise).  The perturbation eliminates the overlap between adjacent edges, and yields an $\alpha$-bend graph with $3n-6$ edges and an $\alpha$-bend multigraph with $4n-6$ edges.
\end{proof}

\begin{corollary}\label{cor:collinear}
For every integer $n\geq 3$, we have
$M_b(n,\alpha)=3n-6$ for all $\alpha\in (0,\pi)$;
and
$M_a(n,\alpha)=3n-6$ for all $\alpha\in [\frac{\pi}{2}, \pi)$.
\end{corollary}

\section{Asymptotics: Large and Small Angles}
\label{sec:LS}

In this section, we study the maximum number of edges in $\alpha$-bend and $\alpha$-arc graphs as the angle $\alpha$ tends to 0 or $\pi$.

\subsection{Large Angles}
\label{ssec:large}
For a sufficiently large $\alpha$, the $\alpha$-bend and $\alpha$-arc edges are similar to straight-line edges.

\begin{proposition}
\label{pro:largeangle}
For every $P\subset \mathbb{R}^2$ in general position, there exists a threshold $\alpha_0\in (0,\pi)$
such that, for every $\alpha\in (\alpha_0,\pi)$, we have $M_a(P,\alpha)=M_b(P,\alpha) = M(P)$ and $M_a^\|(P,\alpha)=M_b^\|(P,\alpha)=2M(P)$.
\end{proposition}
\begin{proof}
Let $G$ be a planar straight-line graph on $P$ with the maximum number of edges (i.e., with $M(P)$ edges). Then $G$ is a triangulation of the convex hull of $P$. Let $\beta$ be the minimum interior angle over all triangular faces of $G$, and $\alpha_1=\pi-\beta$. Recall that the bisectors of the interior angles of a triangle $T$ meet the center $c(T)$ of the inscribed circle of $T$. Subdivide each triangular face $T$ of $G$ into three subtriangles by connecting the corners of $T$ to the center $c(T)$. Then each subtriangle incident to two points in $P$, and contains an $\alpha$-arc edge between them for any $\alpha\in (\alpha_1,\pi)$. The outer face of $G$ is the exterior of $\mathrm{conv}(P)$, and it contains an $\alpha$-arc edge between any two consecutive vertices of $\mathrm{conv}(P)$ for any $\alpha\in (\frac{\pi}{2},\pi)$. Furthermore, these $\alpha$-arc edges are pairwise noncrossing. Overall, we can find $2|E(G)|=2M(P)$ pairwise noncrossing $\alpha$-arc edges, which form an $\alpha$-arc multigraph on $P$. These constructions show that $M_a^\|(P,\alpha)\geq 2M(P)$ and $M_b^\|(P,\alpha)\geq 2M(P)$ for $\alpha \in (\alpha_1,\pi)$.

By choosing an arbitrary bend point in each $\alpha$-arc edge, construct $2|E(G)|=2M(P)$ pairwise noncrossing $\alpha$-bend edges. By deleting double edges, we also obtain $\alpha$-arc and $\alpha$-bend graphs with $|E(G)|=M(P)$ edges. Consequently, $M_a(P,\alpha)\geq M(P)$ and $M_b(P,\alpha)\geq M(P)$  for $\alpha \in (\alpha_1,\pi)$.

For matching upper bounds, let $\alpha_2$ be the maximum angle between any two adjacent straight-line edges determined by $P$. Then, given any $\alpha$-arc graph $HG$ on $P$, the convex hull of any edge $ab$ does not contain any other vertices in $P$. Consequently, we can replace each $\alpha$-arc edge in $H$ with a straight-line edge, and obtain a planar straight-line graph with $|E(H)|$ edges. This proves $M_a(P,\alpha)\leq M(P)$ and $M_b(P,\alpha)\leq M(P)$ for $\alpha \in (\alpha_2,\pi)$. For multigraphs, up to two parallel edges could be replaced by a straight-line edge by Proposition~\ref{pro:multiplicity}, consequently $M_a^\|(P,\alpha)\geq 2M(P)$ and $M_b^\|(P,\alpha)\geq 2M(P)$ for $\alpha \in (\alpha_2,\pi)$.

Overall, we put $\alpha_0=\max\{\alpha_1,\alpha_2\}$, and then for all $\alpha \in (\alpha_0,\pi)$, we have $M_a(P,\alpha)=M_b(P,\alpha) = M(P)$ and $M_a^\|(P,\alpha)=M_b^\|(P,\alpha)=2M(P)$.
\end{proof}

\subsection{Small Angles}
\label{ssec:small}
In the other end of the spectrum, for sufficiently small $\alpha>0$, both $M_b(P,\alpha)$ and $M_b^\|(P,\alpha)$ have the same behavior as for collinear points.

\begin{proposition}\label{pro:small-angle}
For every set $P$ of $n$ points in the plane, there exists a threshold $\alpha_0\in (0,\pi)$ such that for all $\alpha\in (0,\alpha_0)$, we have
$M_b(P,\alpha)=3n-6$ and $M_b^\|(P,\alpha)=4n-6$.
\end{proposition}
\begin{proof}
By applying a rotation, if necessary, we may assume that the points in $P$ have distinct $x$-coordinates.
Let $P=\{p_0,\ldots, p_{n-1}\}$ be sorted by (increasing) $x$-coordinate. Let $\alpha_0$ be the minimum angle between a vertical line and a segment $p_{i-1}p_i$ for $i=1,\ldots , n-1$. Now for any angle $\alpha\in (0,\alpha_0)$, we can follow the argument in the proof of Proposition~\ref{pro:collinear}; see Fig.~\ref{fig:4}~(left).
The $x$-monotone path $(p_0,\ldots , p_{n-1})$ plays the role of the $x$-axis: It separates upper and lower edges.
We initially draw each $\alpha$-bend edge $p_ip_j$ so that its two segments have slopes $\pm \cot \frac{\alpha}{2}$,
and then perturb overlapping edges as in the proof of Proposition~\ref{pro:collinear}.
\end{proof}

\begin{figure}
\centering
\includegraphics[width=0.85\textwidth]{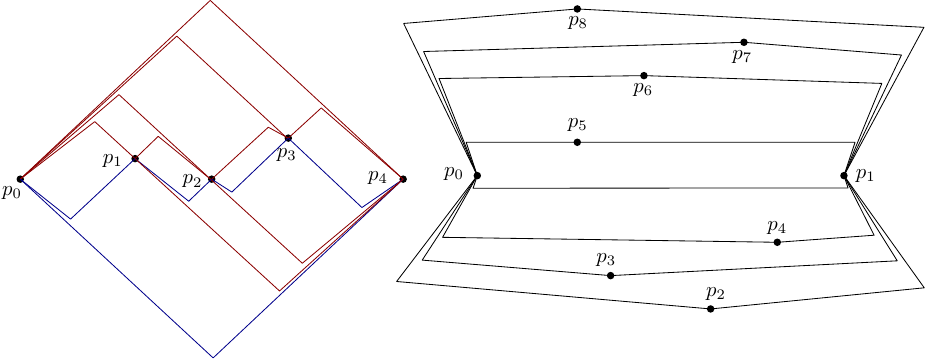}
\caption{Left: an $\alpha$-bend multigraph with $4n-6$ edges; the red edges form an $\alpha$-bend graph with $3n-6$ edges.
Right: an $\alpha$-bend graph with $3n-7=M(P)$ edges for a point set $P$ with a quadrilateral convex hull.} \label{fig:4}
\end{figure}

For points in general position and for $\alpha$-bend graphs, we can take the threshold $\alpha_0$ in Proposition~\ref{pro:small-angle} to be $\frac{\pi}{2}$; see Fig.~\ref{fig:4} (right).
\begin{proposition}
\label{pro:middle}
For every set $P\subset \mathbb{R}^2$ of $n\geq 3$ points in general position, and for every $\alpha\in (0,\frac{\pi}{2}]$,
we have $M_b(P,\alpha)\geq 3n-7$ and $M_b(P,\alpha)\geq M(P)$.
\end{proposition}
\begin{proof}
Let $\{p_0, p_1\}\subset P$ be a diametric pair of points, which maximizes the pairwise distances in $P$. By applying a rotation, if necessary, we may assume that the line segment $p_0p_1$ is horizontal. Since $P$ is in general position, no two points in $P\setminus \{p_0,p_1\}$ have the same $y$-coordinate. Label the points in $P$ such that $p_0$ is the leftmost point, $p_1$ is the rightmost point, and $P\setminus \{p_0,p_1\}=\{p_2,\ldots , p_{n-1}\}$ sorted in increasing $y$-coordinates; see Fig.~\ref{fig:4}~(right).

We construct an $\alpha$-bend graph on $P$ as follows. For all $i\in \{2,\ldots , n-1\}$, add $\alpha$-bend edges $p_0p_i$ and $p_ip_1$ such that the edge segments incident to $p_i$ are horizontal, and lie to the left and right of $p_i$, respectively. Since $\alpha\in (0,\frac{\pi}{2}]$, then the edge segments incident to $p_0$ lie in the closed halfplane left of $p_0$; and the edge segments incident to $p_1$ lies to the right of $p_1$. Note also that the edge segments incident to $p_0$ and $p_1$ may overlap: We perturb these edges to maintain $\alpha$-bend edges, but eliminate the overlap (as a result, the edge segments incident to $p_i$ are no longer horizontal, but almost horizontal). Add the edge $p_0p_1$ as well, where the edge segment incident to $p_0$ is almost horizontal, and the edge segment incident to $p_1$ is very short. We have added $2(n-2)+1=2n-3$ edges so far.

For all $i\in \{2,\ldots, n-3\}$, if both $p_i$ and $p_{i+1}$ are on the same side of the horizontal line $p_0p_1$, then we add an $\alpha$-bend edge $p_ip_{i+1}$ in the horizontal strip between $p_i$ and $p_{i+1}$. The edge segment incident to $p_i$ should almost horizontal (but disjoint from the edges $p_0p_i$ and $p_1p_i$), and on the same side (left or right) of $p_i$ that contains $p_{i+1}$. This determines the direction of the edge segment incident to $p_{i+1}$. There is at most one $i\in \{2,\ldots , n-3\}$ such that $p_i$ and $p_{i+1}$ are on opposite sides of the line $p_0p_1$, so we add at least $n-4$ edges. We obtain an $\alpha$-bend graph with $(2n-3)+(n-4)=3n-7$ edges. This proves that $M(P,\alpha)\geq 3n-8$.

If the convex hull of $P$ contains 4 or more points, then $M(P)\leq 3n-7$, consequently $M(P,\alpha)\geq M(P)$. Suppose that the convex hull of is a triangle. Then $p_0p_1$ is one side of the triangle, and all other points lie on one side of the line $p_0p_1$. In this case, we add edges $p_ip_{i+1}$ for all $i\in \{2,\ldots, n-3\}$, and the we obtain an $\alpha$-bend graph with $(2n-3)+(n-3)=3n-6=M(P)$ edges.
\end{proof}

For circular arcs, the number of edges goes down to 2 as $\alpha$ tends to zero.
\begin{proposition}\label{pro:small-angle-arcs}
For every finite $P\subset \mathbb{R}^2$ in general position, there exists a threshold $\alpha_0\in (0,\pi)$ such that, for every $\alpha\in (0,\alpha_0)$, we have
$M_a(P,\alpha)=1$ and $M_a^\|(P,\alpha)=2$.
\end{proposition}
\begin{proof}
Let $P\subset\mathbb{R}^2$ be a set of $n$ points in general position, let $W$ be the set of intersection points of the $\binom{n}{2}$ lines spanned by $P$.
Let $\varepsilon>0$ be so small that the $\varepsilon$-radius disks centered at the points in $W$ are pairwise disjoint, and let $D$ be a large disk that contains the $\varepsilon$-disks centered at all points in $W$. Finally, let $\alpha_0$ be so small that for every point pair $\{p,q\}\subset P$, the intersection of disk $D$ and the circle containing an $\alpha_0$-arc edge is in the $\varepsilon$-neighborhood of the line spanned by $pq$. Note that the same property holds for $\alpha$-arc edges for all $\alpha\in (0,\alpha_0)$.

For the lower bound, it is clear that an $\alpha$-arc graph can contain any one $\alpha$-arc edge $pq$, and an $\alpha$-arc multigraph can contain any $\alpha$-arc double-edge $pq$.
For the upper bound, suppose to the contrary, that $\alpha\in (0,\alpha_0)$ and a plane $\alpha$-arc graph $G$ contains two $\alpha$-arc edges $p_1q_1$ and $p_2q_2$, where $\{p_1,q_1\}\neq \{p_2,q_2\}$. The circles containing the $\alpha$-arc edges $p_1q_1$ and $p_2q_2$ cross at some point in the $\varepsilon$-neighborhood of the intersection point of lines $p_1q_1$ and $p_2q_2$, which is in $D$. Consequently they also cross at another point, say $x$, outside of $D$. The point $x$ lies on both $\alpha$-arc edges, and so these edges cross, contradicting the assumption that $G$ is a plane graph.
\end{proof}
A combination of Propositions~\ref{pro:collinear} and~\ref{pro:small-angle-arcs} yields the following.
\begin{corollary}\label{cor:small-angle-arcs}
For every finite $P\subset \mathbb{R}^2$, where the maximum number of collinear points is $k$, there exists a threshold $\alpha_0\in (0,\pi)$ such that, for every $\alpha\in (0,\alpha_0)$, we have $M_a(P,\alpha)=2k-3$ and $M_a^\|(P,\alpha)=2k-2$.
\end{corollary}

\section{The Size of $\alpha$-Bend and $\alpha$-Arc Multigraphs}
\label{sec:2pi3}

Let $G$ be an $\alpha$-bend multigraph. The union of two parallel edges between vertices $a$ and $b$ is a closed curve that contains the line segment $ab$ in its interior by Proposition~\ref{pro:multiplicity}. Let $\overline{G}$ denote the straight-line graph comprising a straight-line edge for each double edge in $G$.
A \emph{lens} of $G$ is the interior of the closed curve formed by two parallel edges. A lens is \emph{empty}
if it does not contain any vertex of $G$. Note that if all double edges in $G$ form empty lenses, then $\overline{G}$ is a planar straight-line graph.

\begin{lemma}\label{lem:cycle}
For $\alpha\in (0,\frac{2\pi}{3}]$, let $G$ be an $\alpha$-bend multigraph where every double edge is an empty lens, and let $C=(p_1,\ldots ,p_m)$ be a cycle of double edges. Then the edges of $G$ in the (straight-line) polygon $\overline{C}$ do not contain a triangulation of $C$.  (In particular, $G$ does not contain any 3-cycle of double edges, and there is no diagonal in a 4-cycle of double edges.)
\end{lemma}
\begin{proof}
Suppose, for contradiction, that $C=(p_1,\ldots ,p_m)$ is a cycle of double edges and the edges of $G$ in  the (straight-line) polygon $\overline{C}$ form a triangulation of $C$; see Fig.~\ref{fig:5}. Let $pq$ be an edge in $C$. Then two closed curves pass through $p$ and $q$: An empty lens formed by two $\alpha$-bend edges, and the straight-line cycle $\overline{C}$. These curves can cross only at $p$ and $q$. The straight-line segment $pq$ is in the interior of the empty lens, and all vertices of $\overline{C}$ are in its exterior. Consequently, one of the $\alpha$-bend edges between $p$ and $q$ lies in the interior of $\overline{C}$.

\begin{figure}
\centering
\includegraphics[width=0.85\textwidth]{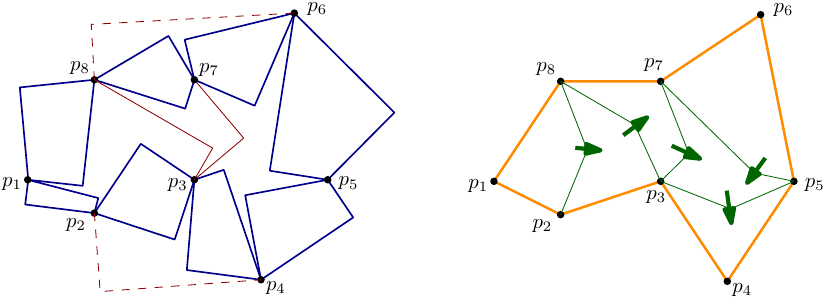}
\caption{Left: a cycle of double edges $C=(p_1,\ldots ,p_m)$, with two internal and two external diagonals.
Right: The corresponding polygon $\overline{C}$, a triangulation of $\overline{C}$ with one-bend edges, and the directions of the dual edges.} \label{fig:5}
\end{figure}

Let $T$ be a dual graph of the triangulation of $\overline{C}$, where the nodes correspond to triangles, and two nodes are adjacent if the corresponding triangles share an edge. The dual graph is a tree $T$. We define an orientation on the edges of $T$ as follows. Consider two adjacent triangles, $t_1$ and $t_2$, that share an edge $p_ip_j$. Direct the edge as $(t_1,t_2)$ if and only if the $\alpha$-bend edge $p_ip_j$ and triangle $t_2$ are on the same halfplane of line $p_ip_j$. Every directed tree contains a sink. Let $t=\Delta(p_i p_j p_k)$ be a sink.

From the discussion above, the straight-line triangle $\Delta(p_i p_j p_k)$ contains three $\alpha$-bend edges: $p_ip_j$, $p_j p_k$, and $p_k p_i$. The concatenation of these edges is a simple hexagon $H$ (formed by three vertices and three bend points). At every bend point the interior angle of $H$ is $2\pi-\alpha> 2\pi-\frac{2\pi}{3}=\frac{4\pi}{3}$. The sum of these three interior angles is greater than $4\pi$. The sum of all interior angles of a hexagon, however, is at most $(6-2)\pi = 4\pi$: a contradiction.
\end{proof}

\begin{lemma}\label{lem:empty}
For $\alpha\in (0,\frac{2\pi}{3}]$, let $G$ be an $\alpha$-bend multigraph on $n\geq 3$ vertices where every double edge is an empty lens. Then $G$ has at most $4n-7$ edges.
\end{lemma}
\begin{proof}
If the double edges do not form any cycle, then there are at most $n-1$ double edges, hence $G$ has at most $(3n-6)+(n-1)=4n-7$ edges, as required. Assume now that $G$ contains a cycle of double edges.

Consider the graph $\overline{G}$ (where each straight-line edge represents a double edge in $G$). By Lemma~\ref{lem:cycle}, $\overline{G}$ is triangle-free. By Euler's polyhedron formula, if $\overline{G}$ has $f$ bounded faces, it has $n+f-1$ edges. That is, $G$ has $n+f-1$ double edges. By Lemma~\ref{lem:cycle}, each bounded face of $\overline{G}$ contains a face of $G$ with 4 or more vertices. By replacing the $n+f-1$ double edges with single edges, and triangulating each face that has 4 or more vertices, we obtain a triangulation $T$. Note that $T$ may have parallel edges in the interior and exterior of a face $f$, but parallel edges cannot form an empty lens, and so $T$ has at most $3n-6$ edges. Consequently, $|E(G)| -(n+f-1)+ f \leq 3n-6$, which yields $|E(G)|\leq 4n-7$, as claimed.
\end{proof}

\begin{theorem}\label{thm:recursion}
For $\alpha\in (0,\frac{2\pi}{3}]$ and $n\geq 2$, we have $M_b^\|(n,\alpha)\leq 4n-6$.
\end{theorem}
\begin{proof}
We proceed by induction on $n$. The base case $n=2$ trivially follows from Proposition~\ref{pro:multiplicity}.
Let $G$ be an $\alpha$-bend graph on $n>2$ vertices, and assume that the theorem holds for every subgraph of $G$ on fewer than $n$ vertices. If all double edges are empty-lenses, then $|E(G)|\leq 4n-7$ by Lemma~\ref{lem:empty}. Otherwise, delete all vertices that lie in the interior of a lens (i.e., a cycle formed by double edges), and let $G'$ be the induced subgraph of the remaining $n'\geq 2$ vertices. Then $G'$ has at most $4n'-7$ edges by Lemma~\ref{lem:empty} if $n'\geq 3$; and at most $2=4n'-6$ by Proposition~\ref{pro:multiplicity} if $n'=2$. In both cases, $G'$ has at most $4n'-6$ edges.

Note that the lenses of a plane multigraph form a laminar system (i.e, two lenses are either disjoint or one contains the other), and so the maximal lenses of $G$ are pairwise disjoint.
If $L_1,\ldots , L_k$ are the maximal lenses that contain $n_1,\ldots n_k$ vertices in their interior, then
$n=n'+\sum_{i=1}^k n_i$. For $i=1,\ldots , k$ let $G_i$ denote the subgraph induced by all vertices inside and on the boundary of $L_i$. The boundary of the lens contains two vertices and two (parallel) edges. In particular $G_i$ has $n_i+2$ vertices.
It has at most $4(n_i+2)-6= 4n_i+2$ edges by induction, but two of these edges are already included in $G'$.
Consequently, $|E(G)|= |E(G')|+\sum_{i=1}^k 4n_i \leq (4n'- 6) + 4\sum_{i=1}^k n_i = 4n-6$, as claimed.
\end{proof}

When $\alpha\in (\frac{2\pi}{3},\pi)$, then $M_b^\|(n,\alpha)$ is close to the trivial upper bound of $6n-12$, established by Proposition~\ref{pro:multiplicity}.
\begin{proposition}\label{pro:grid}
For $\alpha\in (\frac{2\pi}{3},\pi)$, we have $M_b^\|(n,\alpha)\geq 6n-O(\sqrt{n})$.
\end{proposition}
\begin{proof}
Let $n$ points be arranged in a section of a triangular grid. The unit-length edges form a plane graph, with $3n-O(\sqrt{n})$ edges,
where all bounded faces are equilateral triangles. We can replace each unit-length edge by two $\alpha$-bend edges, which form an empty lens, such that the two segments of each edge have equal length. Each equilateral triangle contains three $\alpha$-bend edges, which are crossing-free. This yields a plane $\alpha$-bend multigraph with $6n-O(\sqrt{n})$ edges.
\end{proof}

For $\alpha$-arc multigraphs, Lemma~\ref{lem:empty}, Theorem~\ref{thm:recursion}, and Proposition~\ref{pro:grid} carry over with essentially the same proof, but the angle threshold increases from $\alpha\leq \frac{2\pi}{3}$ to $\alpha\leq \frac{5\pi}{6}$, as a triangle $\Delta(p_1 p_2 p_3)$ cannot contain three $\alpha$-arc edges between its vertices for $\alpha\leq \frac{5\pi}{6}$. We summarize the result and omit the details.

\begin{theorem}\label{thm:arc}
For $\alpha\in (0,\frac{5\pi}{6}]$ and $n\geq 2$, we have $M_a^\|(n,\alpha)\leq 4n-6$.
For $\alpha\in (\frac{5\pi}{6},\pi)$, we have $M_a^\|(n,\alpha)\geq 6n-O(\sqrt{n})$.
\end{theorem}

\section{Applications to $\alpha AC_1^=$ Graphs}
\label{sec:AC}

Ackerman et al.~\cite{AckermanFT12} proved that every $n$-vertex $\alpha AC_1^=$ graph has at most $27n$ edges, for every $\alpha\in (0,\frac{\pi}{2}]$. More precisely, they proved an upper bound of $24.5n$ if $\alpha\neq \frac{\pi}{3}$ and $27n$ if $\alpha=\frac{\pi}{3}$. We improve these bounds to $18.5n$ and $21n$, resp., using the same general strategy combined with Theorem~\ref{thm:recursion} from Section~\ref{sec:2pi3}. For $\alpha=\frac{\pi}{2}$ (i.e., for RAC$_1$-graphs), however, an asymptotically tight bound of $5.5n-O(1)$ is known~\cite{AngeliniBFK20}.
We recall Lemma~2.1 from~\cite{AckermanFT12}.

\begin{lemma}[Ackerman et al.~\cite{AckermanFT12}]\label{lem:segments-old}
Let $\alpha\in (0,\frac{\pi}{2}]$ and let $S$ be a finite set of line segments in the plane such that
any two segments may cross only at angle $\alpha$. Then $S$ can be partitioned into at most three
subsets of pairwise noncrossing segments. Moreover, if $\frac{\pi}{\alpha}$ is irrational or if
$\frac{\pi}{\alpha} =\frac{p}{q}$, where $\mathrm{gcd}(p,q)=1$ and $p$ is even,
$S$ can be partitioned into at most two subsets of pairwise noncrossing segments.
\end{lemma}

The line segments in $S$ may overlap (i.e., intersect in a line segment of positive length). Crossings between line segments is usually defined as follows: Two segments \emph{cross} if they intersect in a single point that lies in the relative interior of both segments. For the application to $\alpha AC_1^=$ graphs, we relax the definition of crossings: Given a set $S$ of line segments and a set $V$ of segment endpoints, two segments in $S$ \emph{cross} if their intersect in a single point that is not in $V$.
The following lemma strengthens Lemma~\ref{lem:segments-old}, and holds under either notion of crossing.
\begin{lemma}
\label{lem:segments}
Let $\alpha\in (0,\frac{\pi}{2}]$ and let $S$ be a finite set of line segments in the plane such that any two segments may cross only at angle $\alpha$. If $\frac{\pi}{\alpha}$ is irrational or if $\frac{\pi}{\alpha} =\frac{p}{q}$, where $\mathrm{gcd}(p,q)=1$ and $2\mid p$, then there exists a subset $S'\subset S$ of pairwise noncrossing segments with $|S'|\geq \frac12\, |S|$. Else $\frac{\pi}{\alpha} =\frac{p}{q}$, where $\mathrm{gcd}(p,q)=1$ and $p=2k+1$ for some $k\in \mathbb{N}$, and then there exists a subset $S'\subset S$ of pairwise noncrossing segments with $|S'|\geq \frac{k}{2k+1}\, |S|$.
\end{lemma}
\begin{proof}
If $\frac{\pi}{\alpha}$ is irrational or if $\frac{\pi}{\alpha} =\frac{p}{q}$, where $\mathrm{gcd}(p,q)=1$ and $2\mid p$, then the claim follows directly from Lemma~\ref{lem:segments-old}.

Assume that $\frac{\pi}{\alpha} =\frac{p}{q}$, where $\mathrm{gcd}(p,q)=1$ and $p=2k+1$ is an odd integer.
The \emph{direction} of a segment $s\in S$, denoted $\mathrm{dir}(s)$, is the minimum counterclockwise angle from the $x$-axis to a line parallel to $s$. Note that $\mathrm{dir}(s)\in [0,\pi)$. Let $D=\{\mathrm{dir}(s): s\in S\}$, that is, the set of directions of the segments in $S$. For each direction $d\in D$, let $S(d)=\{s\in S: \mathrm{dir}(s)=d\}$ be the set of segments of direction $d$.

We define a vertex-weighted graph $G_D = (D, E_D)$, in which two directions $d_1,d_2\in D$ are joined by an edge if and only if they differ by $\alpha$; and the weight of a direction $d\in D$ is the cardinality of $S(d)$.
Clearly, the maximum degree of a vertex in $G_D$ is at most two, and so $G_D$ is the disjoint union of paths and cycles. Furthermore, if $d_1, d_2\in D$ are in the same component of $G_D$, then
they differ by a multiple of $\alpha =\frac{q}{p}\cdot \pi$. Since $\mathrm{gcd}(p,q)=1$,
then $m\cdot \alpha \equiv 0\bmod \pi$ if and only if $m\equiv 0\bmod p$.
Consequently, every cycle in $G_D$ is isomorphic to $C_p=C_{2k+1}$.

Note that if segment $s_1,s_2\in S$ cross, then $\mathrm{dir}(s_1)$ and $\mathrm{dir}(s_2)$ are adjacent in $G_D$.
We can now construct a subset $S'$ of $S$. In each connected component $H$ of $G_D$, we choose a maximum independent set $I(H)$ as follows. (1) If $H$ is a path, then it is 2-colorable, and the weight of one of the color classes is at least half of the weight of $H$; let $I(H)$ be such a color class. (2) If $H$ is isomorphic to $C_{2k+1}$, then it has $2k+1$ maximum independent sets (each containing $k$ vertices), and every vertex lies in precisely $k$ independent sets. By the pigeonhole principle, the weight of an independent set is at least $\frac{k}{2k+1}$ times the weight of $H$; let $I(H)$ be such an independent set. Let $I\subset D$ be the union of independent sets over all components $H$ of $G_D$; and note that $I$ is an independent set in $G_D$, and its weight is at least $\frac{k}{2k+1}$ times the weight of $G_D$, which is $|S|$. Now let $S'$ be the union of the sets $S(d)$ for all $d\in I$. Clearly, the segments in $S'$ are pairwise disjoint and $|S'|\geq \frac{k}{2k+1}\, |S|$, as required.
\end{proof}

\begin{corollary}\label{cor:segments}
Let $\alpha\in (0,\frac{\pi}{2}]$ and let $S$ be a finite set of line segments in the plane such that any two segments may cross only at angle $\alpha$. Then there exists a subset $S'\subset S$ of pairwise noncrossing segments with $|S'|\geq |S|/3$.
\end{corollary}

\begin{lemma}\label{lem:ac}
Let $\alpha \in (0,\frac{\pi}{2}]$, and let $G = (V, E)$ be a graph on $n\geq 3$ vertices that admits an $\alpha AC^=_1$
drawing such that for every edge $e\in E$, both edge segments cross at least one other edge in $E$. Then
\begin{enumerate}\itemsep 0pt
\item[(i)] $|E| \leq 15n-27$ for all $\alpha\in (0,\frac{\pi}{2}]$;
\item[(ii)] $|E| \leq 12n-18$ for $\alpha\in (0,\frac{\pi}{3})$;
\item[(iii)] $|E| \leq 10n-15$ if $\alpha\in (\frac{\pi}{3},\frac{\pi}{2})$; and
\item[(iv)] $|E| \leq 10n-18$ if $\frac{\pi}{\alpha}$ is irrational or $\frac{\pi}{\alpha} =\frac{p}{q}$, where $\mathrm{gcd}(p,q)=1$ and $2\mid p$.
\end{enumerate}
\end{lemma}
\begin{proof}
\textbf{(i)}
Let $S$ be the set of edge segments of all edges in $E$; hence $|S|=2\, |E|$.
For each segment $s\in S$, create a directed path $\gamma(s)$ as follows: Start from the (unique) vertex in $V$ incident to $s$, then follow $s$ until the first crossing with some other segment $s'$, and then follow $s'$ to the (unique) vertex in $V$ incident to $s'$. Each path $\gamma(s)$ is either an $\alpha$-bend edge or a $(\pi-\alpha)$-bend edge. We make three observations:
(1) For every $s\in S$, the first segment of $\gamma(s)$ is crossing free;
(2) if $s_1\neq s_2$, then $\gamma(s_1)\neq \gamma(s_2)$, since the initial segments of the $\gamma$ paths are distinct;
(3) for $s_1\neq s_2$, the paths $\gamma(s_1)$ and $\gamma(s_2)$ may correspond to the same undirected polygonal path with opposite directions.

Let $\Gamma=\{\gamma(s): s\in S\}$, and let $(V,\Gamma)$ be the \emph{red graph}. Let $\Gamma_1\subset \Gamma$ be the set of red edges where the second segment crosses some other edge in $\Gamma$; and let $\Gamma_2=\Gamma\setminus \Gamma_1$ be the set of red edges where both segments are crossing-free.

Two edges in $\Gamma_1$ cannot follow the same path in opposite directions because the first segment of every red edge is crossing-free. Let $S_1$ be the set of second segments of the red edges in $\Gamma_1$. By Corollary~\ref{cor:segments}, there is a subset $S_1' \subseteq S_1$ of pairwise noncrossing segments of size $|S_1'|\geq \frac13\, |S_1|$. Let $\Gamma_1'$ be the set of edges in $\Gamma_1$ whose second segment lies in $S_1'$, with  $|\Gamma_1'|\geq \frac13\, |\Gamma_1|$.
If $\Gamma_2$ contains two edges that follow the same path in opposite directions, then omit one
arbitrarily, and let $\Gamma_2'\subseteq \Gamma_2$ be the remaining edges, with $|\Gamma_2'|\geq \frac12\, |\Gamma_2|$.

Consider the multigraph $(V,\Gamma_1'\cup \Gamma_2')$, and note that $|\Gamma_1'\cup \Gamma_2'|\geq \frac13\, |\Gamma_1| + \frac12\, |\Gamma_2|\geq \min\{\frac13,\frac12\} (|\Gamma_1|+|\Gamma_2|) =\frac13\, |\Gamma|$. The edges in $\Gamma_1'\cup \Gamma_2'$ are pairwise noncrossing but they may overlap. However, if a segment of $\gamma(s)\in \Gamma_1'\cup \Gamma_2'$ overlaps with another edge in $\Gamma_1'\cup \Gamma_2'$, then the other segment of $\gamma(s)$ is overlap-free. Indeed, the first segments of red edges and the edges in $\Gamma_2'$ are pairwise nonoverlapping; and the second segments of edges in $\Gamma_1'$ are pairwise noncrossing (in the strong sense). Consequently, we can partition $\Gamma_1'\cup \Gamma_2'$ into subsets of pairwise overlapping red edges; and we can perturb the edges in each subset to remove overlaps while maintaining the angles between segments. After perturbation, $(V,\Gamma_1'\cup \Gamma_2')$ is a plane multigraph, composed of $\alpha$-bend and $(\pi-\alpha)$-bend edges.

Let $\Gamma_3\subset \Gamma_1'\cup \Gamma_2'$ be the set of $\alpha$-bend edges; and let $\Gamma_4\subseteq \Gamma_1'\cup \Gamma_2'$ be the set of $(\pi-\alpha)$-bend edges. By Theorem~\ref{thm:recursion},
$|\Gamma_3|\leq 4n-6$. In general, the multiplicity of every edge in $(V,\Gamma_4)$ is at most two by Proposition~\ref{pro:multiplicity}, and so $|\Gamma_4|\leq 2(3n-6)=6n-12$. Hence we have $|\Gamma|\leq 3\, |\Gamma_1'\cup \Gamma_2'|= 3\, |\Gamma_3\cup \Gamma_4| \leq 3\big((4n-6)+(6n-12)\big)=30n-54$.
Overall, this yields $|E|=\frac12\, |\Gamma|\leq 15n-27$.

\paragraph{(ii)}
When $\alpha < \frac{\pi}{3}$, we can apply Theorem~\ref{thm:recursion} for both $\alpha$-bend and $(\pi-\alpha)$-bend multigraphs.
Consequently, $|\Gamma_1'\cup \Gamma_2'|=|\Gamma_3|+|\Gamma_4|\leq 2(4n-6)= 8n-12$.
This yields $|\Gamma|\leq 3\, |\Gamma_1'\cup \Gamma_2'|\leq 3(8n-12)=24n-36$, hence $|E|=\frac12\, |\Gamma|\leq 12n-18$.

\paragraph{(iii)}
When $\alpha\in [\frac{\pi}{3},\frac{\pi}{2})$, then $\pi-\alpha<2\pi/3$, and $|\Gamma_4|\leq 4n-6$ by Theorem~\ref{thm:recursion}.
Consequently, $|\Gamma_1'\cup \Gamma_2'|=|\Gamma_3|+|\Gamma_4|\leq 2(4n-6)= 8n-12$.
If $\frac{\pi}{\alpha}=\frac{p}{q}$ with $\mathrm{gcd}(p,q)=1$ and $q=2k+1$ is an odd integer, then $k\geq 2$ and $|S_1'|\geq \frac25 \, |\Gamma_1|$ by Lemma~\ref{lem:segments}.
This yields $|\Gamma|\leq 2.5\, |\Gamma_1'\cup \Gamma_2'|\leq 2.5(8n-12)=20n-30$, hence $|E|=\frac12\, |\Gamma|\leq 10n-15$.

\paragraph{(iv)}
When $\frac{\pi}{\alpha}$ is irrational or $\frac{\pi}{\alpha}=\frac{p}{q}$ with $\mathrm{gcd}(p,q)=1$ and $2\mid p$,
then the size of $S_1'$ is bounded by $|S_1'|\geq \frac12 \, |\Gamma_1|$ by Lemma~\ref{lem:segments}.
This yields $|\Gamma|\leq 2(10n-18)=20n-36$, hence $|E|=\frac12\, |\Gamma|\leq 10n-18$.
\end{proof}

\begin{theorem}\label{thm:AC}
If $G=(V,E)$ is an $\alpha AC_1^=$ graph with $n\geq 3$ vertices,  then
\begin{enumerate}\itemsep 0pt
\item[(i)] $|E| \leq 21n-36$ for $\alpha=\frac{\pi}{3}$;
\item[(ii)] $|E| \leq 18.5n-34$ for $\alpha\in (0,\frac{\pi}{3})$;
\item[(iii)] $|E| \leq 16.5n-31$ if $\alpha\in (\frac{\pi}{3},\frac{\pi}{2}]$; and
\item[(iv)] $|E| \leq 16n-30$ if $\frac{\pi}{\alpha}$ is irrational or $\frac{\pi}{\alpha} =\frac{p}{q}$, where $\mathrm{gcd}(p,q)=1$ and $2\mid p$.
\end{enumerate}
\end{theorem}
\begin{proof}
Let $\alpha\in (0,\frac{\pi}{2}]$; and let $G = (V, E)$ be a graph with $n\geq4$ vertices with an $\alpha AC^=_1$ drawing.
Let $E_1\subseteq E$ denote the set of edges in $E$ that have at least one crossing-free end segment.
Let $G_1 = (V, E_1)$ and $G_2 = (V, E \setminus E_1)$.
By Corollary~\ref{cor:segments}, there is a subset $S_1'\subseteq S_1$ of pairwise noncrossing segments of size $|S_1'|\geq \frac13\,
|E_1|$. The graph $G_1'$ corresponding to these edges is planar, with at most $3n-6$ edges. Hence $E_1$
contains at most $3 \cdot (3n-6) = 9n-18$ edges.
By Lemma~\ref{lem:ac}(i), $G_2$ has at most $15n-18$ edges.
Hence, $G$ has at most $(9n-18)+(15-18)=24n-36$ edges in general.
However, we can improve on this bound for all $\alpha\in (0,\frac{\pi}{2}]$.

\paragraph{(i)}
When $\alpha=\frac{\pi}{3}$, we can use Lemma~\ref{lem:ac}(ii),
which yields at most $(9n-18)+(12n-18)=21n-36$ edges.

\paragraph{(ii)}
When $\alpha\in (0,\frac{\pi}{2})$ and $\alpha\neq \frac{\pi}{3}$, then $G_1$ is quasi-planar (i.e., it does not contain three pairwise crossing edges). Indeed, suppose to the contrary that three edges in $G_1$ in pairwise cross. As every edge in $G_1$ has two segments, one of which is crossing-free, then these edges contain three segments that pairwise cross. These segments form a triangle, in which every interior angle is $\alpha$ or $\pi-\alpha$. Since interior angles of a triangle sum to $\pi$, it follows that $\alpha=\frac{\pi}{3}$; a contradiction.

It is known~\cite{Ackerman2020,AckermanT07} that a simple quasi-planar $n$-vertex graph has at most $6.5n-20$ for $n\geq 4$, hence at most $6.5n-16$ edges for $n\geq 3$. Combined with Lemma~\ref{lem:ac}(i), $G$ has at most $18.5n-34$ edges.

\paragraph{(iii)}
When $\alpha\in (\frac{\pi}{3},\frac{\pi}{2})$, then again $G_1$ is quasi-planar, with at most $6.5n-16$ edges due to~\cite{Ackerman2020,AckermanT07}. Combined with Lemma~\ref{lem:ac}(iii), $G$ has at most $16.5n-31$ edges.

\paragraph{(iv)}
When $\frac{\pi}{\alpha}$ is irrational or $\frac{\pi}{\alpha} =\frac{p}{q}$, where $\mathrm{gcd}(p,q)=1$ and $2\mid p$.
Then we can assume $|S_1'|\geq \frac12\, |E_1|$ by Lemma~\ref{lem:segments}, hence $|E_1|\leq 2\cdot (3n-6) = 6n-12$.
Combined with Lemma~\ref{lem:ac}(iv), $G$ has at most $16n-30$ edges in this case.
\end{proof}

\section{Conclusions and Open Problems}
\label{sec:con}

We have introduced $\alpha$-bend and $\alpha$-arc graphs and multigraphs for any $\alpha\in (0,\pi)$, and derived upper and lower bounds on the maximum number of edges, $M_a(P,\alpha)$ and $M_b(P,\alpha)$, for a point set $P$ in these drawing styles. However, the computational complexity of the corresponding optimization problems is unknown.

\begin{problem}\label{prob:1}
Is it NP-hard to determine $M_a(P,\alpha)$ (resp., $M_b(P,\alpha)$) for a given point set $P\subset \mathbb{R}^2$ and angle $\alpha\in (0,\pi)$? Is it $\exists \mathbb{R}$-hard to determine $M_b(P,\alpha)$?
\end{problem}

Intuitively, allowing one-bend edges instead of straight-line edges gives extra flexibility. In several cases, we have shown that an $\alpha$-bend graph on a point set $P$ has at least as many edges as a straight-line triangulation on $P$, that is, $M_b(P,\alpha)\geq M(P)$. However, we have also shown that $\alpha$-arc edges become obstacles as $\alpha$ tends to zero, and $M_a(P,\alpha)< M(P)$ for all sufficiently small $\alpha>0$.
For $\alpha$-bend edges in general, this remains an open problem.

\begin{problem}
Does there exist a finite point set $P\subset \mathbb{R}^2$ and an angle $\alpha\in (0,\pi)$ such that $M_b(P,\alpha)<M(P)$ ?
\end{problem}

It is also an open problem to improve the upper or lower bounds for the number of edges in $\alpha AC_1^=$ graphs and $\alpha AC_2^=$ graphs.

 \begin{problem}
Determine the maximum number of edges in an $n$-vertex $\alpha AC_1^=$ graph and an $n$-vertex $\alpha AC_2^=$ graph for all angles $\alpha\in (0,\frac{\pi}{2}]$.
\end{problem}

\bibliographystyle{alphaurl}
\bibliography{angles}

\end{document}